\documentclass{article}
\usepackage[letterpaper]{geometry}

\def\withcolors{1}
\def\withnotes{1}

\def\BibTeX{{\rm B\kern-.05em{\sc i\kern-.025em b}\kern-.08em
    T\kern-.1667em\lower.7ex\hbox{E}\kern-.125emX}}

\usepackage[utf8]{inputenc} 
\usepackage{hyperref} 
\usepackage{url} 
\usepackage{amsfonts} 
\usepackage[noend]{algcompatible}
\usepackage[noend]{algpseudocode}
\usepackage{soul}
\usepackage{mleftright}
\usepackage{varwidth}
\usepackage[utf8]{inputenc} 
\usepackage[T1]{fontenc} 
\usepackage{hyperref}
\usepackage{url} 

\usepackage[usenames,dvipsnames,table]{xcolor}
\usepackage{amsmath}
\usepackage{algorithm}
\usepackage[noend]{algpseudocode}
\algsetblock[Name]{Encoder}{Stop}{3}{1cm}
\algsetblock[Name]{Decoder}{Stop}{3}{1cm}
\algsetblock[Name]{SharedRandomness}{Stop}{3}{1cm}

\usepackage[noadjust]{cite}
\usepackage{%
    amssymb,%
    amsthm,
    bbm,%
    etoolbox,%
    mathtools,%
    stmaryrd%
}

\usepackage{tikz,pgf,pgfplots,circuitikz}

\theoremstyle{plain} \newtheorem{thm}{Theorem}[section]
\newtheorem{lem}[thm]{Lemma} 
\newtheorem{cor}[thm]{Corollary}

\theoremstyle{definition} \newtheorem{defn}[thm]{Definition}

\theoremstyle{remark} \newtheorem{rem}{Remark}

\definecolor{lightgray}{gray}{0.9}

\newcommand{\eqdef}{:=}

\newcommand{\norm}[1]{\|#1\|}%
\newcommand\numberthis{\addtocounter{equation}{1}\tag{\theequation}}
 \newcommand{\R}{\mathbb{R}}
\newcommand{\N}{\mathbb{N}}
\newcommand{\bPr}[1]{\mathbb{P}\left(#1\right)}
\newcommand{\E}[1]{\mathbb{E}\left[#1\right]}

\newcommand{\C}{\mathcal{C}}

\newcommand{\X}{\mathcal{X}}
\newcommand{\bX}{\mathbb{X}}

\newcommand{\U}{\mathcal{U}}

\newcommand{\uRoman}[1]{\uppercase\expandafter{\romannumeral#1}}

\usepackage{subcaption}
\DeclareCaptionSubType*{algorithm}

\DeclareCaptionLabelFormat{alglabel}{Alg.~#2}

\ifnum\withcolors=1
  \newcommand{\mcolor}[1]{{\color{ForestGreen}#1}} 
  \newcommand{\tcolor}[1]{{\color{Orange}#1}} 
\else

  \newcommand{\mcolor}[1]{{#1}}
  
  \newcommand{\tcolor}[1]{{#1}}
\fi

\ifnum\withnotes=1
  \newcommand{\mnote}[1]{\par\mcolor{\textbf{M: }\sf #1}} 
  \newcommand{\tnote}[1]{\par\tcolor{\textbf{T: }\sf #1}} 
  
\else
  \newcommand{\mnote}[1]{}
  \newcommand{\tnote}[1]{}
  
\fi
\newcommand{\ignore}[1]{\leavevmode\unskip} 

\newcommand{\indic}[1]{\mathbbm{1}_{#1}}

\usepackage{soul}

\providecommand*\email[1]{\href{mailto:#1}{#1}}

\newcommand{\cE}{\mathcal{E}}
\newcommand{\cO}{\mathcal{O}}

\newcommand{\cN}{\mathcal{N}}

\newcommand{\mutualinfo}[2]{ I\mleft(#1 \land #2\mright) }
\newcommand{\SNR}{\mathtt{SNR}}

\input{tangocolors.sty}

\usepackage{pstricks,pifont}
\usepackage{tikz}
\usepackage{pgfplots,pgfplotstable}
\usepackage{amsmath}
\allowdisplaybreaks

\definecolor{red1}{rgb}{0.4,0,0}

\definecolor{red1}{rgb}{0.4,0,0}
\hypersetup{
colorlinks,
 citecolor=blue,
  linkcolor=red1,
  urlcolor=Blue
}

\usepackage{lipsum}

\newcommand\blfootnote[1]{%
  \begingroup
  \renewcommand\thefootnote{}\footnote{#1}%
  \addtocounter{footnote}{-1}%
  \endgroup
}

\begin{document}

\title{Fundamental limits of  over-the-air optimization: Are analog schemes optimal?}
\author{Shubham K Jha \thanks{Robert Bosch Center for Cyber-Physical Systems, Indian Institute of Science, Bangalore, India. } \and 
     Prathamesh Mayekar\textsuperscript{$\dag$} \and Himanshu Tyagi\textsuperscript{$\ast$} \thanks{Department of Electrical Communication Engineering,
Indian Institute of Science, Bangalore, India. Email: \email{
{\{shubhamkj,  prathamesh, htyagi\}@iisc.ac.in }}}}
\maketitle 

\begin{abstract}
We consider over-the-air convex optimization
on a $d-$dimensional space
where coded gradients are sent over an additive Gaussian noise channel
with variance $\sigma^2$.
The codewords satisfy an average power constraint $P$,  resulting in the signal-to-noise ratio (SNR)
of $P/\sigma^2$. 
We derive bounds for the convergence rates for over-the-air optimization.
Our first result is a lower bound
for the convergence rate showing that any code must {slowdown}
the convergence rate
by a factor of roughly $\sqrt{d/\log(1+\SNR)}$.
Next, we consider  a popular class of schemes called \emph{analog coding}, where a linear function of the gradient is sent.  
We show that a simple scaled transmission analog coding scheme results in a slowdown in convergence rate by
a factor of $\sqrt{d(1+1/\SNR)}$. This matches the previous lower
bound up to constant factors for low SNR, making the scaled transmission scheme optimal
at low SNR. However, we show that this slowdown is necessary for any analog coding scheme.
In particular, a slowdown in convergence by a factor of $\sqrt{d}$ for analog coding remains even
when SNR tends to infinity.
Remarkably, we present a simple quantize-and-modulate scheme that uses \emph{Amplitude Shift Keying}
and almost attains the optimal convergence rate at all SNRs.
\end{abstract}

\blfootnote{{An abridged version of this paper will appear in the proceedings of  IEEE Global Communications Conference (GLOBECOM), 2021.}}
\date{}
\newpage
\tableofcontents
\newpage

\section{Introduction}
Distributed optimization is a classic topic with decades of work
building basic theory.
The last decade has seen increased interest in this topic
motivated by distributed and large scale machine learning. For instance,
parallel implementation of training algorithms for
deep learning models 
over multi-GPU has become commonplace.
In another direction, over the past 5 years
or so,
federated learning applications that require
building machine learning models for data distributed across multiple
users have motivated optimization algorithms that limit communication
from the users to a parameter server ($cf.$~\cite{konevcny2016federated}).
Most recently, there has been a lot of interest in the scenario
where this communication is {\em over-the-air}, namely the users
are connected over a wireless communication channel
($cf.$~\cite{Amiria19, Amiri19}).

Many different
optimization algorithms
have been proposed
using different kinds of codes. However, there is no work addressing
information-theoretic limits on the performance of these algorithms.
In particular, it remains unclear whether simple analog schemes for
communication
over AWGN channel are optimal in any setting and whether there is any
fundamental limitation to their performance. More broadly, do we still
need sophisticated error-correcting codes to attain the optimal
convergence rate for the optimization problem? In this work, we
address these questions for convex optimization problems. 

We establish an information-theoretic lower bound on the convergence rate for any scheme for
convex stochastic optimization, which shows that, for $d-$dimensional domain,
there is a $\displaystyle{\sqrt{\frac{d}{\log (1+\SNR)}}}$ factor slowdown in convergence rate.
Furthermore, for low $\SNR$, analog codes with stochastic gradient descent (SGD)
attain this optimal rate. Next, we establish a general lower bound on
the performance of analog codes and show that
there is a factor $\sqrt{d(1+\frac{1}{\SNR})}$ slowdown in convergence rate when
analog codes are used. Note that as $\SNR$ goes to infinity
one can expect that the convergence rate should tend to the classic
one. But our bound shows that for analog codes 
there is at least a factor $\sqrt{d}$ slowdown
 even as the $\SNR$ tends to infinity, making them suboptimal at high $\SNR$. 
 Finally, we show that a simple
quantize-and-modulate SGD
scheme that uses a vector quantizer for the gradients and sends the
quantized values using amplitude shift keying (ASK) is almost rate
optimal.

There has been a very interesting line of work on these topics,
including \cite{Amiria19, Amiri19, Amirib20, Amiric19,  Abad20,  Chang20,Sery19,Sery20,Kobi20,Yang20, Zhu20, Zhang21, Zhu21, Wang18, Chen21, Sun20, saha2021decentralized}.
Most works have considered the multiparty setting,
with more complicated channels than AWGN.
In this paper, for simplicity, we restrict to the two-terminal
setting. But our qualitative results apply to the multiparty setting
as well. 

Broadly, the gradient coding schemes proposed in these works  can be
divided into two categories: { analog} and { digital}. In more
detail, in analog schemes, the coded gradients sent over the noisy channel are a linear transformation of the subgradient supplied by the oracle.  Typical analog schemes include scaling,  sparsification, or direct transmission of gradients over a wireless channel.  For instance,  authors in \cite{Amiria19} send only top $k$ gradient coordinates along with error feedback. In \cite{Kobi20}, the subgradient estimates are scaled-down appropriately to satisfy the power constraint. Each coordinate is then transmitted over the Gaussian channel using one channel use per transmission.  Similar scaling approaches are also presented in \cite{Sery20,  Sun20, Yang20, Zhang21, Zhu21}.
 On the other hand,  digital schemes rely on gradient quantization and channel coding. For instance, authors in \cite{Chang20} propose to quantize the subgradients using stochastic quantization, and the precision is chosen so that the transmission rate is the same as channel capacity. Then they are transmitted using any capacity-achieving code. 
 In \cite{Zhu21}, authors perform one-bit quantization of subgradients similar to signSGD \cite{bernstein18a} and send them over-the-air using OFDM modulation, taking into account the frequency selective-fading and inter-symbol interference.

 In summary, most of the prior work
   either uses analog schemes or capacity-achieving channel codes.
   Further, even works such as~\cite{Zhu21} which use a quantize-and-modulate approach like our work,
do not comment on the optimality of the rate of convergence. In fact, 
in our proposed scheme, we use a one-dimensional signal constellation and let the number of bits 
   used for quantization grow roughly as $\log(1+\SNR)$ to get optimal convergence rate.

 In a slightly different direction, the variant of distributed optimization with compressed subgradient estimates has also been studied extensively, primarily to mitigate the slowdown in convergence of distributed optimization procedures when full gradients are communicated (see, for instance,   \cite{alistarh2017qsgd, gandikota2019vqsgd, basu2019qsparse, faghri2020adaptive, seide20141, wang2018atomo, wen2017terngrad, acharya2019distributed,    mayekar2020ratq, lin2020achieving, lin2020achieving, mayekar2020limits, suresh2017distributed, chen2020breaking, huang2019optimal, mayekar2021wyner, jhunjhunwala2021adaptive, ghosh2020distributed}).

 We build on the quantizers proposed in these works to obtain a nearly optimal convergence rate algorithm.

 For our lower bounds, we follow a similar strategy as~\cite{ACMT21informationconstrained}
(which in turn builds on~\cite{ACT:18,agarwal2012information, acharya2020general, ach2020disc})
 where optimization under communication constraints  (not over-the-air) was considered. While
 the difficult oracles of these prior works yield our general lower bound, for deriving the limitation for analog schemes, we consider a new class of Gaussian oracles; see Section~\ref{s:proof} for more details.

The rest of the paper is organized as follows. We set up the problem in
the next section and provide all our main results in
Sections~\ref{s:results}. All the proofs are given in
Section~\ref{s:proof}
and concluding remarks are in Section~\ref{s:conclusion}.


\section{Problem formulation and preliminaries}
\subsection{Functions and gradient oracles}
For a convex set $\X\subset \R^d$ with $\sup_{x,y\in \X}\|x-y\|\leq D$,
we consider the minimization of an unknown convex function $f:\X\to\R$
using access to a first order {\em oracle} $O$ that reveals
noisy subgradient estimates for any queried point. 
We assume that the oracle outputs $\hat{g}(x)$
when a point $x\in \X$ is queried satisfy
the following conditions:
\begin{align}
 \E {\hat{g}(x)|x} &\in \partial f(x), \quad\text{(unbiasedness)}
\label{e:unbiasedness}
\\
\E{\|\hat{g}(x)\|^2|x} &\leq B^2, \quad \text{(mean square bounded oracle)}
\label{e:mean_square}
\end{align}
where $\partial f(x)\subset \R^d$ denotes the set of subgradients of $f$ at input
$x$. Denote by $\cO$ the set of pairs $(f,O)$
of functions and oracles satisfying the conditions above.

\subsection{Codes and Gaussian channel}\label{s:probF}
In our setting, the gradient estimates are not directly available to
the optimization algorithm $\pi$ but must be coded for error
correction, sent over a noisy channel, and decoded
 to be used by $\pi$. We consider
fixed length codes of length $\ell$ with average power less than
$P$. Specifically, we consider $(d, \ell, P)$-codes consisting
of encoder mappings $\varphi: \R^d \times \mathcal{U} \to \R^\ell$
such that
the codeword $\varphi(\hat{g},U)\in \R^\ell$ used to send the
subgradient estimate 
$\hat{g}\in \R^d$
satisfies the average power constraint
\begin{align}\label{e:Power_constraint}
\E{\|\varphi(\hat{g},U)\|^2}\leq \ell P,
\end{align}
where $U\in \U$ denotes the public randomness used to randomize the encoder
and is assumed to be available to both $\varphi$ and optimization algorithm $\pi$. For convenience, we drop the argument $U$ from the notation of $\varphi$ for the rest of the paper.
Denote by $\C_{\ell}$ the set of all $(d, \ell, P)$-codes.

After the $t$th query by the algorithm, when the oracle supplies a subgradient estimate $\hat{g}_t,$ the codeword $C_t=\varphi(\hat{g}_t)$ is sent over an {\em additive Gaussian
  noise}
channel. That is, after the $t$th query to the oracle, the algorithm $\pi$ observes $Y_t \in \R^\ell$ 
given by
\begin{align}\label{e:Gaussian_channel}
Y_t(i)=C_t(i)+Z_t(i), \qquad
1\leq i \leq \ell,
\end{align} where $\{Z_t(i)\}_{i\in [\ell], t \in \N}$ is a sequence of i.i.d. random variables
with common distribution $\mathcal(0, \sigma^2)$ -- the Gaussian distribution with mean $0$
and variance $\sigma^2$. We denote the {\em signal-to-noise ratio} by
$\displaystyle{
\SNR := \frac{P}{\sigma^2}.
}$

\subsection{Over-the-air Optimization}
We now describe an optimization algorithm $\pi$ using $(d,
\ell,P)$-code $\varphi$.
In any iteration $t$,  the optimization algorithm $\pi$,
upon observing the previous channel outputs $Y_1, ...,
Y_{t-1}\in\R^\ell$, queries the oracle
with point\footnote{We assume that the downlink
communication  channel from the algorithm to the oracle is
noiseless.} $x_t$. The oracle
gives $\hat{g}_t\in \partial f(x_t)$, encodes it as $\varphi(\hat{g})$
and sends it over the Gaussian channel.
The algorithm $\pi$ observes the output $Y_t\in \R^\ell$ of the
channel and moves to iteration $t+1$.

After $T$ iterations, the algorithm outputs $x_T$. Denote by $\Pi_{\ell, T}$ the class of all algorithms using a $(d, \ell, P)$-code and making $T$ oracle queries.

We abbreviate the
overall algorithm $\pi$ with access to oracle $O$ and using encoder
$\varphi$ by $\pi^{\varphi O}$. We call the tuple $(\pi, \varphi)$ consisting of the optimization algorithm and the encoding procedure $\varphi$ as an \emph{over-the-air optimization protocol.}
The convergence error of this over-the-air optimization protocol is given by
\[\cE(f, \pi^{\varphi O}):=\E {f(x_T)}-\min_{x\in \X}f(x).\]

We want to study how the convergence error goes to zero as a function of the total number of
channel-uses $N=T\ell$. We are allowed to use codes with any length
$\ell$ but note that an increase in the length of encoding protocol will lead to a decrease in the number of oracle queries as the number of channel-uses is restricted to $N$. Similarly, while we are allowed to use an optimization algorithm that can make as many as $N$ queries to the oracle, increasing the number of queries will lead to a smaller block length encoding protocol. Let
$\Lambda(N):=\{\pi \in \Pi_{\ell, T}, \varphi \in \C_{\ell}\colon \ell \cdot T \leq N\}$.
That is, $\Lambda(N)$ is the set of all over-the-air optimization protocols using $N$ channel transmissions.
Then, the smallest worst-case convergence error possible by using $N$
channel transmissions is given by
$\cE^\ast(N, \X):= \inf_{(\pi,\varphi) \in \Lambda(N)}\sup_{(f,O)\in \cO}
\cE(f, \pi^{\varphi O})$.
Let $\bX :=\{\X\colon \sup_{x,y \in \X}\norm{x-y}\leq D\}.$
In this paper, we will characterize the following quantity\footnote{Our goal behind considering
the min-max cost in
  \eqref{e:minmax} is to ensure that the lower bounds are independent of the geometry of set $\X$. But our upper bound techniques can handle an arbitrary, fixed $\X$ as well.}:
\begin{align}\label{e:minmax}
\cE^\ast(N):=\sup_{\X \in \bX }\cE^\ast(N, \X).
\end{align}

\subsection{Special coding schemes}
In addition to the general coding scheme above, we are interested in
the following
two special classes of simple coding schemes: Analog codes and ASK codes.

\begin{defn}\label{d:analog_schemes}
A code is an \emph{analog code} if the encoder mapping $\varphi$ is
  linear, i.e., when $\varphi(x)=\mathbf{A}x$ for an $\ell \times d$ matrix $\mathbf{A}$, for any $\ell \leq d$.
 We allow for random matrices $\mathbf{A}$ as long as they are independent of the observed gradient estimates. Also, we denote by $\cE^\ast_{analog}(N)$ the min-max optimization error when the class of $(d, \ell, P)$-encoding protocol is restricted to analog schemes (with everything else remaining the same as in \eqref{e:minmax}). Clearly,  $\cE^\ast_{analog}(N) \geq \cE^\ast(N).$ 
\end{defn}

\begin{defn}\label{d:ASK_codes}
  A code is an\footnote{For simplicity, we have considered AWGN channel for transmission.
    In many practical communication systems, a two-dimensional signal space is available through the in-phase
    and quadrature-phase components. For these systems, our results for ASK code
continue to hold with a QAM or QPSK constellation-based code.} \emph{Amplitude Shift Keying (ASK)
  code} satisfying the average power constraint \eqref{e:Power_constraint} if the range
of the encoder mapping is given by
\[
\left\{-\sqrt{P} +\frac{(k-1)\cdot 2\sqrt{P}}{2^r-1}\colon k \in [2^r] \right\},
\]
for some $r\in \N$.
Namely, the encoder first quantizes $\hat{g}$ to $r$ bits and then uses ASK
modulation for sending the quantized subgradient estimate. Note that
this is a code of length $1$.
\end{defn}

\subsection{A benchmark from prior results}
We recall results for the case $\SNR=\infty$, namely
the classic case when gradients estimates supplied by the oracle
are directly available to $\pi$, since perfect decoding is possible for every channel-use. We denote the min-max error in this case by $\mathcal{E}^{\ast}_{classic}{(N)}$. In
this standard setup for first-order convex optimization, prior work
gives a complete characterization of the min-max error
$\mathcal{E}^{\ast}_{classic}{(N)}$; see,
for instance, \cite{ nemirovski1995information}.
We summarize these well-known results below.
\begin{thm}
For absolute constants $c_1 \geq c_0 >0$, we have
\[
\displaystyle{  \frac{c_0DB}{ \sqrt{N} } \leq \mathcal{E}^{\ast}_{classic}{(N)}
\leq \frac{c_1DB}{ \sqrt{N} } }
.
\]
\end{thm}
Thus, the $1/\sqrt{N}$ convergence rate that  SGD provides for convex functions is optimal up to
constant factors, with dependence on the dimension $d$ coming only through the parameters $D$
and $B$. This convergence rate will serve as a basic benchmark for our results in this paper.

\section{Main Results}\label{s:results}
\subsection{Lower Bound for over-the-air optimization}
We begin by proving a lower bound for over-the-air optimization. The proof of the lower bound uses recent results in  information-constrained optimization given in \cite{ACMT21informationconstrained}, which in turn builds on the results of \cite{agarwal2012information, acharya2020general}.
As is usual in other lower bounds in stochastic optimization,
our lower bound holds for a sufficiently large $N$. 
\begin{thm}\label{t:LB}
 For some universal constant\footnote{The universal constants differ in different theorem statements.} $c \in (0, 1)$ and  $N \geq \frac{d}{\log(1+\SNR)}$, we have\footnote{$\log(\cdot)$ and $\ln(\cdot)$ denote logarithms to the base $2$ and  base $e$, respectively.}
 \[
 \cE^\ast(N) \geq \frac{cDB}{\sqrt{N}}\cdot \sqrt{\frac{d}{\min\{d, 1/2\log(1+\SNR)\}}}.
 \]
\end{thm}
Our lower bound states that there is slowdown by  a factor of $\sqrt{\frac{d}{\log(1+\SNR)}}$
over the classic convergence rate and no over-the-air optimization scheme can achieve the classic convergence rate unless the $\SNR$ is sufficiently high. 

\subsection{Performance and limitations of analog schemes}
Next, we show that a simple analog coding scheme attains the optimal convergence rate at low $\SNR$. Specifically, we consider the scheme from~\cite{Kobi20} where
the subgradient estimate is scaled-down appropriately to satisfy the power constraint in \eqref{e:Power_constraint}, sent coordinate-by-coordinate over $d$ channel-uses, and then scaled-up before using
it in a gradient descent procedure.
We call this analog code the {\em scaled transmission} analog code.
Throughout the paper, our first-order optimization algorithm remains projected subgradient descent algorithm (PSGD), with different codes and associated decoding schemes to
get back the transmitted subgradient estimate. 
\begin{thm}\label{t:AUB}
The over-the-air optimization procedure $(\pi, \varphi)$ comprising
  the scaled transmission analog code and PSGD satisfies
  \[
  \sup_{(f,O)\in \cO}
  \cE(f, \pi^{\varphi O})  \leq \frac{cDB}{\sqrt{N}} \cdot \sqrt{d+\frac{d}{\SNR}},  
  \]
  where $c$ is a universal constant. 

\end{thm}
Since $\sqrt{d+({d}/{\SNR})}\leq \sqrt{2d/\SNR}\leq \sqrt{3d/\log (1+\SNR)}$
for a sufficiently small $\SNR$, we get the following corollary in view of Theorem~\ref{t:LB}
and the result above. 

\begin{cor}\label{c:AnOpt}
There exist universal constants $c_1, c_2$ such that  for $\SNR \in (0,1)$ (i.e., low $\SNR$) and $N \geq \frac{d}{\log (1+\SNR)},$ we have
  \begin{align*}
\frac{c_1DB}{\sqrt{N}} \cdot
\sqrt{\frac{d}{\log(1+\SNR)}}\leq  \cE^\ast_{analog}(N)
\leq \frac{c_2DB}{\sqrt{N}} \cdot
\sqrt{\frac{d}{\log(1+\SNR)}}.
  \end{align*}
\end{cor}  
\begin{rem}
We remark that a slightly different analog coding scheme can also guarantee the same performance as the scaled transmission analog code given in Theorem \ref{t:AUB} and performs better in our experiments presented in Section \ref{s:exp}. In this scheme, the noisy subgradient estimate is first randomly rotated by a random matrix, and then only a few of its coordinates are used for the gradient descent procedure, which, in turn,  are sampled randomly.  Both the random matrix and random coordinate sampling are generated using shared randomness between the encoder and the algorithm. Notice that such an algorithm needs only few channel-uses per descent step instead of scaled transmission analog code that uses $d$ channel-uses per descent step.
We provide a detailed description and analysis of this scheme in Section \ref{s:scheme1}.
\end{rem}

Interestingly, our next result shows that the scaled transmission scheme is the optimal analog coding scheme up to constant factors. 
In particular, while analog codes are optimal for low $\SNR$, they can be far from optimal at high $\SNR$. 
\begin{thm}\label{t:ALB}
For some universal constant $c \in (0, 1)$  and $N \geq d(1+1/\SNR)$, we have 
\[
\cE^\ast_{analog}(N) \geq \frac{cDB}{\sqrt{N}}\cdot \sqrt{d+\frac{d}{\SNR}}.
\]
\end{thm}

Note that for small $\SNR$,  we have $1+\frac{1}{\SNR}\approx \frac{1}{\log(1+\SNR)}$, and thus, Theorem \ref{t:ALB} shows that analog codes are optimal at low $\SNR$. Theorem \ref{t:ALB}  also shows that in comparison to Theorem \ref{t:LB} analog schemes can lead to a slowdown of $\sqrt{d}$ for high values of $\SNR$. Even when $\SNR$ goes to infinity, we can't get the classic, dimension-free convergence rate back. Note that the upper bound in Theorem~\ref{t:AUB} matches the lower bound of Theorem~\ref{t:ALB} for large $\SNR$, establishing that the scaled transmission analog code of~\cite{Kobi20} is optimal among analog coding schemes even at high $\SNR$. We remark that the convergence analysis in~\cite{Kobi20} required additional smoothness assumptions and is not valid for our setting.

\begin{rem}\label{r:topk}
While our definition of analog schemes does not include the top-$k$ (see, for instance, \cite{alistarh2018convergence} and the references therein) analog coding schemes, we can also derive a lower bound for such schemes.  Even for such analog schemes,
similar lower bound as above holds and 
  the convergence rate does not match the classic convergence rate at high $\SNR$. We defer the details to the Appendix \ref{ap:TopK}.
\end{rem}

\subsection{Optimality of ASK}\label{s:optASK}
We now present a code that almost attains the convergence rate in the lower bound of Theorem~\ref{t:LB}.
Our encoder $\varphi$ quantizes the noisy subgradient estimates by using a \emph{gain-shape} quantizer ($cf.$\cite{gersho2012vector}). 
\begin{defn}[Gain-shape quantizer]\label{Gain-shape quantizer}
A Quantizer Q is defined to be a gain-shape quantizer if it has the following form
\[Q(Y)=Q_g(\norm{Y}_2)\cdot Q_s(Y/\norm{Y}_2),\]
where $Q_g$ is any $\R \rightarrow \R$ quantizer and $Q_s$ is any $\R^d \rightarrow \R^d$ quantizer.
\end{defn}
 That is, the encoder separately quantizes the norm of the subgradient, its {\em gain}, and the normalized vector obtained after dividing the subgradient by its norm, its {\em shape}.
The quantized gain and shape are sent over two different channel-uses using ASK code. We note
that this scheme is not strictly an ASK code since we use the channel twice. However, this is just a technicality and can be avoided by a more tedious analysis.

To clearly present our ideas, we first present an ASK code which works in 
an ideal setting,
captured by the  following assumptions for the quantized subgradient:
\begin{enumerate}
\item (Perfect gain quantization) We assume that the norm of subgradient vector can be perfectly sent to the algorithm i.e.,  without any induced noise.  Further, we don't account for the channel-uses in sending the norm.
\item (An ideal shape quantizer) There exists an ideal shape\footnote{We call this an ideal quantizer because such a quantizer would achieve the lower bound for stochastic optimization in \cite{mayekar2020ratq}, where the gradients are quantized to $r$-bits.} quantizer which quantizes the shape of the vector to a mean square error of $d/r$ and where the quantized output is an unbiased estimate of the input. 
\end{enumerate}
Recall that our optimization algorithm is PSGD with an appropriate decoding rule to decode the noisy codewords sent over the channel.
\begin{thm}\label{t:ASK}
  Under Assumptions 1-2 above, there exists an over-the-air optimization procedure $(\pi, \varphi)$ with
  an ASK code $\varphi$ for which we have
\[ \sup_{(f,O)\in \cO}
\cE(f, \pi^{\varphi O})  \leq \frac{2DB}{\sqrt{N}} \cdot \sqrt{\frac{d}{{\min\lbrace d,\log \left(\sqrt{\frac{4\SNR}{\ln N}}+1\right)\rbrace}}}.
\]
Furthermore, the ASK code quantizes the subgradient vector
to ${r= \log \left(\sqrt{\frac{4\SNR}{\ln N}}+1\right)}$ bits.
\end{thm}
\begin{rem}[Resolution grows with $\SNR$]
  We remark that the number of bits $r$ used to express the subgradients in our algorithm
    grows with $\SNR$ as
     ${r= \log \left(\sqrt{\frac{4\SNR}{\ln N}}+1\right)}$ bits, namely the resolution must grow logarithmically with $\SNR$.
\end{rem}

We now state our complete result, without making ideal assumptions.  This time the gain is sent in one channel-use by simply scaling the gain value appropriately to satisfy the power constraint, which is similar to the scaled transmission analog code from Theorem \ref{t:AUB}.  For quantizing the shape, our scheme uses the quantizer RATQ from \cite{mayekar2020ratq}. 
Again note above,  this scheme is not formally an ASK code since we send the gain over a separate channel. Nonetheless, they are similar, in essence, to ASK codes as only the transmission of gain, a scalar, is not accounted for in the ASK code.

\begin{thm}\label{t:ASK_1}
  For  $d$, $\SNR$, and $N$ satisfying\footnote{$\ln^*a$ denotes the smallest number of $\ln$ operations on $a$ required to make it less than 1. Also, we remark that we can prove a similar convergence bound without any upper bound on $\ln^* (d/3)$; we only make this assumption to simplify the upper bound expression.}  $\ln^\ast (d/3) \leq 7$ and
$\log \left(\sqrt{\frac{4\SNR}{\ln N}}+1\right) \geq 6$, we have
\[
\cE^\ast(N) \leq \frac{2DB}{\sqrt{N}} \cdot \sqrt{\frac{d}{\min\{d, \frac{r}{48}\}}},
\]
where $\displaystyle{r= \log \left(\sqrt{\frac{4\SNR}{\ln N}}+1\right)}$.
Furthermore, this bound is attained by using an over-the-air optimization procedure consisting of PSGD as the optimization algorithm and an ASK-like encoding procedure.\footnote{In particular, the encoding procedure uses two channel-uses for transmitting the subgradient estimate. In the first channel-use, an ASK code is used to transmit the shape of the subgradient vector, which is quantized to $r= \log \left(\sqrt{\frac{4\SNR}{\ln N}}+1\right)$ bits. In the second channel-use, the gain of the subgradient vector is transmitted after scaling it appropriately to satisfy the power constraint.}
\end{thm}

\section{Proofs}\label{s:proof}
We first prove our lower bounds before coming to the algorithms and upper bounds. 

\subsection{ Summary of the method used for proving lower bounds}
 We follow the recipe of \cite{ACMT21informationconstrained} to prove our lower bounds.
The difficult functions we construct are the same as in
previous lower bounds for convex functions such as \cite{agarwal2012information}. We consider the domain $\X=\{x \in \R^d:\norm{x}_\infty\leq D/(2\sqrt{d} \}$, and consider the following class of functions on $\X$: For $v \in \{-1, 1 \}^d$, let
\begin{equation}
  \label{eq:def:gv:convex:p2}
  f_{v}(x) \eqdef \frac {2B\delta}{\sqrt{d}} \sum_{i=1}^{d} \mleft|x(i) - \frac{v(i) D}{2 \sqrt{d}}\mright|, \quad \forall\, x\in \X.
\end{equation}
Note that the gradient $g_v(x)$ of $f_v$ at $x\in \X$ is equal to $-2B\delta v /\sqrt{d}$,
i.e., it is independent of $x$. We will fix our noisy subgradient oracle $O_v$ later.
For any $O_v$, let $\hat{g}_t$ denote the output of the gradient oracle in iteration $t$. We will
consider a noisy oracle which outputs $\hat{g}_t$ that are i.i.d. from a distribution $p_v$
with mean $-2B\delta v /\sqrt{d}$.

For a given code $\varphi$ of length $\ell$, let $C_t=\varphi(g_t)$, $t=1, ...,T$. 
Let $V \sim Unif \{-1, 1\}^d$ and $Y^T=(Y_1, \ldots, Y_T)$
denote output of the AWGN channel when the inputs are $C^T=(C_1, \ldots, C_T)$. 
The following lower bound can be established by using results from\footnote{Note that the result in \cite{ACMT21informationconstrained} is for a more general class of adaptive channels.} \cite[Lemma 3, 4]{ACMT21informationconstrained}:
\begin{equation}\label{e:C_P2}
  \E{f_{V}(x_T)-f_{V}(x_V^\ast)} \geq {\frac{DB\delta}{6}}  \Bigg[1 -  \sqrt{
\frac{2}{d}
      \sum_{i=1}^d \mutualinfo{V(i)}{Y^T}}\Bigg].
\end{equation}
By the definition of $\cE^\ast(N)$, 
we have 
\begin{align} \label{e:minmax_sub-optimality}
  \cE^\ast(N) \geq  \E{f_{V}(x_T)-f_{V}(x_V^\ast)}.
\end{align}
Thus, it only remains to bound the mutual-information term. Note that this
bound holds for any oracle $O_v$;
we choose difficult oracles satisfying ~\eqref{e:unbiasedness} and~\eqref{e:mean_square}
to derive our lower bounds.

\subsection{Proof of Theorem \ref{t:LB}}
\paragraph{A difficult gradient oracle}
For each $f_v$ in \eqref{eq:def:gv:convex:p2}, consider a gradient oracle $O_v$ which
outputs $\hat{g}_t$ with
independent coordinates, each taking values
$-B/\sqrt{d}$ or $B/\sqrt{d}$ with probabilities $(1+2\delta v(i))/2$ and $(1-2\delta v(i))/2$, respectively. The parameter $\delta >0$ is to be chosen suitably later.
Note that  $\hat{g}_t$ are product Bernoulli distributed vectors with mean $-2B\delta v/\sqrt{d}$.
\paragraph{Bounding the mutual-information}
The following strong data processing inequality was derived in~\cite{acharya2020general}
for $\mutualinfo{V(i)}{Y^T}$ when the observations are product Bernoulli vectors:
\[
\sum_{i \in [d]}\mutualinfo{V(i)}{Y^T}\leq  c^{\prime}\delta^2 \max_{v \in \{-1,1\}^d}\max_{\varphi \in \C_{\ell}}\mutualinfo{\hat{g}^T}{Y^T},
\]
where $c^\prime$ is some constant.
Using the well-known formula for AWGN capacity (see \cite{CovTho06}),
we can show using the data processing inequality that 
\[\sum_{i \in [d]}\mutualinfo{V(i)}{Y^T}\leq  c^{\prime}\delta^2 N \min \{d, 1/2 \log(1+\SNR) \}.\]
The proof is completed by combining this bound
with \eqref{e:C_P2} and \eqref{e:minmax_sub-optimality},
and maximizing the right-side of \eqref{e:C_P2} by setting $\delta{=}\sqrt{d/(4 c^\prime N \min\{2d, \log(1+\SNR ))}$.

\subsection{Proof of Theorem \ref{t:ALB}} 
Consider the encoder $\varphi(\hat{g}_t)=\mathbf{A}\hat{g}_t$ corresponding to an analog coding scheme
for the functions
in \eqref{eq:def:gv:convex:p2}.
\paragraph{Gaussian oracle} For every $f_v$ and any query point $x_t$,  consider a Gaussian oracle that outputs: $\hat{g}(x_t){=}-2B\delta v/\sqrt{d}+G$,  where $G \sim \cN(0,B^2/d\mathbf{I}_d)$.
For matrix $\mathbf{A}\in \R^{\ell\times d}$,  the subgradients are encoded as $\varphi(\hat{g}(x_t))=\mathbf{A}\hat{g}(x_t)$ and sent over the Gaussian channel. 
\paragraph{Bounding the mutual-information}
We proceed as in the previous lower bound proof and first
note that $\sum_{i=1}^d I(V(i)\wedge Y^T)
\leq I(V\wedge Y^T)$ since $V(i)$ are i.i.d. Further, 
since $Y_1, .., Y_T$ are i.i.d. conditioned on $V$,
we have
$I(V\wedge Y^{T})\leq TI(V\wedge Y_1)$.  Thus,  it suffices to  bound the mutual information $I(V\wedge Y_1)$ which we do in the following lemma. Recall that the outputs $C_t=(-2B\delta/\sqrt{d})\mathbf{A}V+\mathbf{A}G$, where
$G$ denotes the Gaussian noise of the oracle,
satisfies the power constraint $\sum_{t=1}^T\E{\|C_t\|_2^2}\leq T\ell P$,
which implies that $\mathrm{Tr}(\mathbf{A}\mathbf{A}^{\top})B^2/d\leq \ell P/(1+4\delta^2)$.
Further,
$Y_t = C_t+Z_t$, where $Z_t$ is the channel noise in $\ell$ uses. 
\begin{lem} For $\mathbf{A}, G, V$ and $Y_t$
defined above, if  $\frac{\mathrm{Tr}(\mathbf{A}\mathbf{A}^{\top})B^2}{d}\leq \frac{\ell P}{1+4\delta^2}$,  then $\forall t\in [T],$ \[I(V\wedge Y_t)\leq (2\log e)\cdot\ell\delta^2( 1+1/\SNR)^{-1}.\]
\end{lem}
\begin{proof}
Since $C_t=(-2B\delta/\sqrt{d})\mathbf{A}V+\mathbf{A}G$, we have
$\E{C_tC_t^\top}=\left(\frac{B^2(1+4\delta^2)}{d}\right)\mathbf{A} \mathbf{A}^\top$ which implies
\begin{align}\label{e:pow}
\E{\|C_t\|^2}=\mathrm{Tr}(\E{C_tC_t^\top})=\left(\frac{B^2(1+4\delta^2)}{d}\right)\mathrm{Tr}(\mathbf{A} \mathbf{A}^\top)\leq \ell P.
\end{align}
As $Z_t$ is independent of $C_t$, we also have
\begin{align}
\E{Y_tY_t^\top}
&=\left(\frac{B^2(1+4\delta^2)}{d}\right)\mathbf{A} \mathbf{A}^\top+\sigma^2\mathbf{I}_{\ell}.
\end{align}
By definition of mutual-information and the fact that Gaussian maximizes the entropy,
\[I(V\wedge Y_t)= h(Y_t)-h(Y_t | V)\leq\frac{1}{2}\log \frac{\mathrm{det}\left(\frac{B^2(1+4\delta^2)}{d}\mathbf{A} \mathbf{A}^\top+\sigma^2\mathbf{I}_{\ell}\right)}{\mathrm{det}\left(\frac{B^2}{d}\mathbf{A} \mathbf{A}^\top+\sigma^2\mathbf{I}_{\ell}\right)}.\] Let $\lambda_1,\dots, \lambda_\ell$ be the eigen values of $\mathbf{A}\mathbf{A}^\top$.  Then, right-side can be further bounded as
\begin{align*}
I(V\wedge Y_t)
&\leq \frac{1}{2}\sum_{i=1}^\ell \log \frac{\frac{B^2(1+4\delta^2)}{d}\lambda_i+\sigma^2}{\frac{B^2}{d}\lambda_i+\sigma^2}\\
&\leq \frac{\ell}{2} \log \frac{\frac{B^2(1+4\delta^2)}{\ell d}\mathrm{Tr}(\mathbf{A}\mathbf{A}^\top)+\sigma^2}{\frac{B^2}{\ell d}\mathrm{Tr}(\mathbf{A}\mathbf{A}^\top)+\sigma^2}\\
&\leq \frac{(2 \log e)\cdot\ell\delta^2}{1+\ell d\sigma^2/(B^2\mathrm{Tr}(\mathbf{A}\mathbf{A}^\top))}\\
&\leq \frac{(2\log e)\cdot \ell\delta^2}{1+\SNR^{-1}},
\end{align*}
where the first inequality is the Hadamard inequality; the second is Jensen's inequality; the third one uses $\log(1+x)\leq x \log e;$ and the last inequality follows from \eqref{e:pow}.
\end{proof}
\noindent Combining the previous bound with \eqref{e:C_P2} and maximizing the right-side using $\delta {=} \sqrt{(1+\frac{1}{\SNR})\frac{d}{((\log e)\cdot 16N)}}$,  the proof is completed using \eqref{e:minmax_sub-optimality}.

\subsection{A general convergence bound for over-the-air optimization}
For an $\ell$-length coding scheme $\varphi:\R^d \to \R^\ell$, recall that the overall output of the channel $Y_t$ after the $t$th query is given by \eqref{e:Gaussian_channel}.
Our proposed schemes in Sections \ref{s:scheme1} and \ref{s:scheme2} below involve  projecting back this channel output in $\R^\ell$ to $\R^d$.
In particular,  as a part of the optimization algorithm $\pi$, $Y_t$ is passed through
a {\em decoder mapping} $\psi: \R^\ell \to \R^d$ which gives back a $d-$dimensional vector to be used by the first-order optimization algorithm.

 We use PSGD as the first-order optimization algorithm; the overall over-the-air optimization procedure is described in Algorithm \ref{a:SGD_Q}.  PSGD proceeds as SGD, with the additional projection step where it projects the updates back to domain $\X$ using the map $\Gamma_{\X}(y):=\min_{x\in \X}\|x-y\|$, $\forall\, y\in \R^d$. 
\begin{figure*}[h]
\centering
\begin{tikzpicture}[scale=1, every node/.style={scale=1}]
\node[draw,text width= 5 cm , text height= ,] {%
\begin{varwidth}{\linewidth}       
            \algrenewcommand\algorithmicindent{0.2em}
\begin{algorithmic}[1]\label{a:SGD_Q}
   \For{$t=0$ to $T-1$}
   \State ~Observe $Y_t$ given by \eqref{e:Gaussian_channel}
	\State ~$x_{t+1}=\Gamma_{\X} \left(x_{t}-\eta \psi(Y_t)\right)$
   \EndFor \State Output $\frac 1 T \cdot {\sum_{t=1}^T x_t}$
\end{algorithmic}
\end{varwidth}};
 \end{tikzpicture}
 \vspace{-0.1cm}
 \renewcommand{\figurename}{Algorithm}
\caption{Over-the-air PSGD with encoder $\varphi$, decoder $\psi$}\label{a:SGD_Q}
\vspace{-0.1cm}
\end{figure*}

We now derive a convergence bound for over-the-air optimization described in Algorithm \ref{a:SGD_Q}. In our formulation, the decoder $\psi$ is a part of the optimization protocol $\pi$.
However,
for concreteness,
with a slight abuse of notation 
we now denote the overall over-the-air optimization protocol using
the tuple $(\pi, \varphi, \psi)$.
The performance of  $(\pi, \varphi, \psi)$  is controlled by the worst-case $L_2$-norm $\alpha(\pi,\varphi,\psi)$  and the worst-case bias $\beta(\pi,\varphi, \psi)$ of the subgradient obtained after processing the received vector, defined below:
\vspace{-0.25cm}
\begin{align}
&\alpha(\pi,\varphi, \psi):=\sup_{\substack{\hat{g} \in \R^d: \E {\|\hat{g}\|^2}\leq B^2}} \sqrt{\E {\| \psi(Y)\|^2}}, \label{def_a}\\
&\beta(\pi,\varphi, \psi):= \sup_{\substack{\hat{g}\in \R^d: \E {\|\hat{g}\|^2}\leq B^2}}\|\E {(\hat{g}-\psi(Y)}\|, \label{def_b}
\vspace{-0.25cm}
\end{align}
where for all $i\in [d]$, $Y(i)$ satisfies \eqref{e:Gaussian_channel}. The next result is only a minor modification of the standard PSGD proof and is very similar to  \cite[Theorem \textcolor{red1}{2.4}]{mayekar2020ratq}.
\begin{lem}
For the above PSGD equipped over-the-air optimization protocol $(\pi,\varphi, \psi)$ with $N$ channel-uses, we have
\[\sup_{(f,O)\in \cO}
\cE(f, \pi^{\varphi O})  \leq D\left( \frac{\alpha(\pi,\varphi,\psi)}{\sqrt{N/\ell}}+\beta(\pi,\varphi,\psi) \right),
\]
provided that the learning rate $\eta_t$ is set to $\frac{D}{\alpha(\pi,\varphi, \psi)\sqrt{N/\ell}}$ for all iterations $t \in [N/\ell]$.
\label{thm1} 
\end{lem}
This general convergence bound will be used in our upper bound proofs below.

\subsection{Proof of Theorem \ref{t:AUB}} \label{s:scheme1}
 \subsubsection{The scaled transmission analog scheme}
\paragraph{Downscale the power} 
The subgradient vector is multiplied by  $\sqrt{Pd}/B$ to meet the power constraints and sent using $d$ channel-uses,  one channel-use per coordinate. Thus, our encoded output is $\varphi(\hat{g}(x_t)) = \sqrt{Pd}/B \cdot\hat{g}(x_t).$

\paragraph{Upscale the power} The optimization algorithm $\pi$ observes $Y_t$ given by  \eqref{e:Gaussian_channel} and re-scales it back by a factor $B/\sqrt{Pd}$. Thus, the decoding $\psi$ rule at the algorithm's end is given by
$\psi(Y_t)=B/\sqrt{Pd}Y_t$.  It is easy to see that  $\E{\psi(Y_t) | x_t}=\E {\hat{g}(x_t)| x_t}$ implying $\beta(\pi, \varphi, \psi)=0$. Also,  using the independence of zero mean noise $Z_t$ and $\hat{g}(x_t)$, $\E { \|\psi(Y_t)\|^2| x_t}=\E {\|\hat{g}(x_t)\|^2+ B^2/(Pd)\|Z_t\|^2 }$ which can be bounded by
$B^2+B^2\sigma^2/P$.  That implies $\alpha(\pi, \varphi, \psi)\leq B\sqrt{(1+1/\SNR)}$ and the proof is completed using Lemma~\ref{thm1}.

 \subsubsection{The sampled version of scaled transmission analog scheme}\label{s:panlgschm}
\paragraph{Rotate randomly} At each iteration $t$, the subgradient  vector is rotated by multiplying it with a random matrix \[\mathbf{R}:=\frac{1}{\sqrt{d}}\mathbf{HD},\] where $\mathbf{H}$ is a $(d\times d)-$ Walsh-Hadamard matrix \cite{horadam2012hadamard}\footnote{We assume that $d$ is a power of 2.} and $\mathbf{D}$ is diagonal matrix  with each non-zero entry generated uniformly from $\{+1, -1\}.$ The diagonal matrix is generated via public randomness between the encoder and the algorithm, and can therefore be used for decoding at the algorithms end.\\
From \cite[Lemma 5.8]{mayekar2020ratq}, each coordinate $\mathbf{R}\hat{g}_t(i)$ of the rotated subgradient $\mathbf{R}\hat{g}_t$ satisfies
\[\E{\mathbf{R}\hat{g}_t(i)^2} \leq \frac{B^2}{d}, \quad \forall i \in [d].\]

\paragraph{Subsampling} Using shared randomness between the encoder and the decoder, a set $S\subseteq [d]$ is sampled uniformly over all subsets of $[d]$ of cardinality $\ell$. The rotated subgradient vector is sampled at $S$ and is denoted as
\[\tilde{g}_{\mathbf{R},S,t}:=\sum_{i=1}^d\mathbf{R}\hat{g}_t(i) \mathbbm{1}_{\{i\in S\}} \cdot e_i.\]

\paragraph{Downscale the power} 
The subsampled vector $\tilde{g}_{\mathbf{R},S,t}$ is multiplied by  $\sqrt{Pd}/B$ to meet the power constraints and sent using $\ell$ channel-uses, one channel-use per coordinate. Thus, the encoded output is $\varphi(\hat{g}(x_t)) = \sqrt{Pd}/B \cdot\tilde{g}_{\mathbf{R},S,t}.$

\paragraph{Upscale the power} The optimization algorithm $\pi$ observes $Y_t$ given by  \eqref{e:Gaussian_channel} and re-scales it back by a factor $\frac{dB}{\ell\sqrt{Pd}}\cdot \mathbf{R}^{-1}$. Thus, the decoding $\psi$ rule at the algorithm's end is given by
$\psi(Y_t)=\frac{dB}{\ell\sqrt{Pd}}\cdot \mathbf{R}^{-1}Y_t$.  It is easy to see that  
\begin{align*}
\E{\psi(Y_t) | x_t}&=\E{\frac{d}{\ell}\cdot \mathbf{R}^{-1}\left( \sum_{i=1}^d\mathbf{R}\hat{g}_t(i) \mathbbm{1}_{\{i\in S\}} \cdot e_i  \right)| x_t}\\
&=\E{\hat{g}_t\E{\frac{d}{\ell} \mathbbm{1}_{\{i\in S\}}}|x_t}\\
&=\E {\hat{g}(x_t)| x_t},
\end{align*}
 implying $\beta(\pi, \varphi, \psi)=0$. Also,  using the independence of AWGN noise $Z_t\sim \cN(0,\sigma^2 \mathbf{I}_\ell)$ and $\hat{g}(x_t)$,  we get
 \begin{align*}
 \E { \|\psi(Y_t)\|^2| x_t}&=\frac{d^2}{\ell^2}\E {\|\hat{g}(x_t)\|^2\mathbbm{1}_{\{i\in S\}}|x_t}+ \frac{dB^2}{P\ell^2}\|Z_t\|^2\\
 &\leq \frac{B^2 d}{\ell}\left(1+\frac{\sigma^2}{P} \right),
 \end{align*}
 which implies $\alpha(\pi, \varphi, \psi)\leq B\sqrt{\frac{d}{\ell}\left(1+\frac{1}{\SNR}\right)}$, and the proof is complete using Lemma~\ref{thm1}.
\subsection{Proof of Theorem \ref{t:ASK}} \label{s:scheme2}
Since an ASK code is of length $1$,  we can have $N$ queries in $N$ channel-uses.  For the minimum-distance decoder $\psi$,  denote by $A_N$ the event where all the ASK constellation points
sent in $N$ channel-uses
are decoded correctly by the algorithm and by $A_N^c$ as its complement, i.e.,
\[A_N^c := \bigcup_{t=1}^{N} \left\{|Z_t| \geq \frac{2\sqrt{P}}{(2^r-1)}\right\},\]
where $Z_t$ is defined in \eqref{e:Gaussian_channel}.
By the assumptions about an ideal quantizer ($c.f.$ Section \ref{s:optASK}), under the event $A_N$, which depends only
on the channel noise,  Lemma~\ref{thm1} with $\alpha(\pi, \varphi, \psi)=\sqrt{d/r}$
gives \[\E{(f(x_T) -f(x^\ast) )\indic{A_N}}\leq \frac{DB}{\sqrt{N}}\cdot \sqrt{\frac{d}{r}}.\]
Further, due to Gaussian\footnote{In fact, the proof requires noise to be only sub-Gaussian, a weaker assumption than being Gaussian. } noise, we have
{$\bPr{A_N^c} \leq N \exp( - \frac{2P}{\sigma^2 (2^r-1)^2})= N \exp( - \frac{2\SNR}{(2^r-1)^2})$.} Setting {$r=\log \left(\sqrt{\frac{4\SNR}{\ln N}}+1\right)$},  we have $\bPr{A_N^c} \leq \frac{1}{\sqrt{N}}$, which leads to
\begin{align*}
\E{(f(x_T) -f(x^\ast) )}&= \E{(f(x_T) -f(x^\ast) )\indic{A_N}}+\E{(f(x_T) -f(x^\ast) )\indic{A_N^c}}\\
&\leq  \frac{DB}{\sqrt{N}}\cdot \sqrt{\frac{d}{r}}+ \frac{DB}{\sqrt{N}}\\
&\leq 2 \frac{DB}{\sqrt{N}}\cdot {\sqrt{\frac{d}{\min\{d, r\}}}}.
\end{align*}
\subsection{Proof of Theorem \ref{t:ASK_1}}\label{p:ASK_1} For communication, we consider an ASK code in $[-\sqrt{P},\sqrt{ P}]$ with the following $2^r-$constellation points
\begin{align}\label{e:ASK}
\left\{-\sqrt{P} +\frac{(i-1)\cdot 2\sqrt{P}}{2^r-1}\colon i \in [2^r] \right\},
\end{align}
for some $r\in \N.$

We separately send the \textit{gain} $\|\hat{g}_t\|\in \R$ and \textit{shape} $\frac{\hat{g}_t}{\|\hat{g}_t\|}\in \R^d$ of subgradient $\hat{g}_t$.  The encoder $\varphi$ is a tuple which consists of separate gain and shape encoders $\varphi_g$ and $\varphi_s$, i.e.,  $\varphi=( \varphi_g, \varphi_s)$.  The internal randomness used in these gain and shape encoders will be independent, which will result in the output of these encoders being conditionally  independent given any subgradient estimate $\hat{g}_t$.  Similarly,  the decoding mechanism is also a tuple consisting  of two separate decoders $\psi_g$ and $\psi_s$ that are used to decode the transmitted gain and shape values, respectively.  {{The final decoded output of $\psi$ is taken to be the product of the decoded outputs of $\psi_g$ and $\psi_s$.}}

\paragraph{1. Communicating the gain.}

Recall that the gain sub-encoders above need to satisfy the average power constraint \eqref{e:pow} from Section \ref{s:probF}. 

The gain $\|\hat{g}_t\|$ is multiplied by  $\sqrt{P}/B$ to meet the power constraints and sent in one channel-use.  Thus, our encoded output is $\varphi_g(\|\hat{g}_t\|)=(\sqrt{P}/B) \|\hat{g}_t\|.$ The optimization algorithm $\pi$ observes the channel output $Y_{g,t}$ given by
 \[Y_{g,t}=\varphi_g(\|\hat{g}_t\|)+Z_{g,t},\]
 where $Z_{g,t}\sim \cN(0,\sigma^2)$ denotes the Gaussian noise,  and re-scales it back by a factor $B/\sqrt{P}$, i.e.,
 \[\psi_g(Y_{g,t})=\|\hat{g}_t\|+B/\sqrt{P}\cdot Z_{g,t}.\]
 We evaluate the performance measures $\alpha(\pi, \varphi_g,\psi_g)$ and $\beta(\pi, \varphi_g,\psi_g)$,  viewing the gain $\|\hat{g}_t\|$ as a 1-dimensional subgradient.   Similar to \eqref{def_a}, \eqref{def_b} we have, 
\vspace{-0.25cm}
\begin{align*}
&\alpha(\pi,\varphi_g, \psi_g):=\sup_{\substack{\|\hat{g}_t\| \in \R: \E {\|\hat{g}_t\|^2}\leq B^2}} \sqrt{\E {\| \psi_g(Y_{g,t})\|^2}}, \\
&\beta(\pi,\varphi_g, \psi_g):= \sup_{\substack{\|\hat{g}_t\|\in \R: \E {\|\hat{g}_t\|^2}\leq B^2}}\|\E {\|\hat{g}_t\|-\psi_g(Y_{g,t})}\|.
\vspace{-0.25cm}
\end{align*} 
Specifically, it is easy to see that 
 \begin{align}\label{e:perf1}
 \alpha(\pi, \varphi_g,\psi_g)=\sqrt{B^2+B^2/\SNR}, \ \ \  \beta(\pi, \varphi_g,\psi_g)=0.
 \end{align}

\paragraph{2. Quantizing the shape.} We denote the shape $\hat{g}_t/\|\hat{g}_t\|$ by $\hat{g}_{t,\mathtt{shape}}$.  In every iteration $t$, $L_2-$norm of $\hat{g}_{t,\mathtt{shape}}$ is {almost surely} bounded by 1.  Accordingly, to quantize the shape, we are interested in quantizers for {almost surely} bounded oracles. {We use a subsampled version of   RATQ~\cite[Section 3.5]{mayekar2020ratq}  to quanitze the shape, as this quantizer is almost optimal for communication-constrained optimization with almost surely bounded oracles. } The encoder $\varphi_s$ is composed of four components: {\em rotation},  {\em subsampling},  {\em tetra-iterated adaptive quantization} and {\em mapping to ASK code}, which we describe below. 

\subparagraph{Rotation.} Assuming that $d$ is a power of 2,  the subgradient shape $\hat{g}_{t,shape}$ is rotated by multiplying it with a random matrix \[\mathbf{R}:=\frac{1}{\sqrt{d}}\mathbf{HD},\] where $\mathbf{H}$ is a $(d\times d)-$ Walsh-Hadamard matrix \cite{horadam2012hadamard} and $\mathbf{D}$ is a diagonal matrix  with diagonal entries generated uniformly from $\{+1, -1\}.$ The diagonal matrix is generated via public randomness between the encoder and the algorithm, and can therefore be used for decoding at the algorithms end.

Note that since $\mathbf{R}$ is a unitary matrix, the norm remains unaltered even after rotation, i.e., $\|\mathbf{R}\hat{g}_{t,\mathtt{shape}}\|=\|\hat{g}_{t,\mathtt{shape}}\|=1$ a.s..

\subparagraph{Subsampling.} Using shared randomness between the encoder and the decoder,  a set $U \in [d]$ is sampled uniformly over all subsets of $[d]$ of cardinality $\mu d$.  {{The rotated shape vector is sampled at $U$  and is denoted by}} \[\tilde{g}_{\mathbf{R},U,t}:=\{\mathbf{R}\hat{g}_{t,\mathtt{shape}}(i)\}_{i\in U}.\] 
We now quantize every coordinate of $\tilde{g}_{\mathbf{R},U,t}\in \R^{\mu d}$ using the following.

\subparagraph{Tetra-iterated Adaptive Quantization.} 
Consider a sequence of intervals $\{[-M_i,M_i]\}_{i\in [h_s]}$ where $M_1, \dots, M_{h_s}$ grows using\footnote{The $i$th tetra-iteration of $e^{\ast i}$ is defined as: $e^{\ast 1}=e,  e^{\ast i}:=e^{e^{\ast (i-1)}}$.} tetra-iteration:
\[M_1^2=\frac{3}{d},  M_i^2=\frac{3}{d}\cdot e^{\ast (i-1)},  \  i\in [h_s],\]
where parameter $h_s$ satisfies $\log h_s=\lceil[\log (1+\ln^\ast(d/3))\rceil$.  We choose these values such that the largest interval must contain $\|\tilde{g}_{\mathbf{R},U,t}\|_{\infty}$, i.e., $1\leq M_{h_s}.$

For each coordinate $\tilde{g}_{\mathbf{R},U,t}(i)$,  the quantizer first identifies the smallest index $j\in [h_s]$ such that $|\tilde{g}_{\mathbf{R},U,t}(i)|\leq M_{j}$ and then represent $\tilde{g}_{\mathbf{R},U,t}(i)$ using a uniform $k_s$-level shape quantizer $Q_{M_{j}, k_s}$ in interval $[-M_j,M_j]$.  The $k_s$ levels of shape quantizer are given by
\[B_{M_j,k_s}(l):=-M_j+(l-1)\cdot \frac{2M_{j}}{k_s-1}, \ l\in [k_s].\] These levels partition $[-M_j,M_j]$ into $k_s-1$ sub-intervals $\{[B_{M_j}(l), B_{M_j}(l+1)]\}_{l\in [k_s-1]}$.  The uniform quantizer locates a sub-interval that contains $\tilde{g}_{\mathbf{R},U,t}(i)$, say $[B_{M_j}(l^{\ast}), B_{M_j}(l^{\ast}+1)]$ for some $l^{\ast}\in [k_s-1]$, and outputs
\[Q_{M_j, k_s}(\tilde{g}_{\mathbf{R},U,t}(i))=\begin{cases}
B_{M_j}(l^{\ast}), & \hbox{w.p. } \  \frac{B_{M_j}(l^{\ast}+1)-\tilde{g}_{\mathbf{R},U,t}(i)}{B_{M_j}(l^{\ast}+1)-B_{M_j}(l)}\\
\vspace{5pt}
B_{M_j}(l^{\ast}+1), &\hbox{w.p.}  \ \frac{\tilde{g}_{\mathbf{R},U,t}(i)-B_{M_j}(l^\ast)}{B_{M_j}(l^\ast+1)-B_{M_j}(l)}
\end{cases}.\]
This is done for all $\mu d$ coordinates and we represent the output as $Q(\tilde{g}_{\mathbf{R},U,t})$ given by \[Q(\tilde{g}_{\mathbf{R},U,t}):=(Q_{M_{j_{(i)}}}(\tilde{g}_{\mathbf{R},U,t}(i)):  1\leq i\leq \mu d),\] where $j_{(i)}$ corresponds to the index identified for $i$th coordinate $\tilde{g}_{\mathbf{R},U,t}(i).$

Note that there is no overflow because of the choice of $h_s$ and the quantized output $Q(\tilde{g}_{\mathbf{R},U,t})$ can be represented using precision of at most $\mu d\cdot(\log(k_s)+\log(h_s))$ bits.  We denote this binary representation by $[Q(\tilde{g}_{\mathbf{R},U,t})]_2.$

\subparagraph{Mapping to ASK code.} Using the ASK code in \eqref{e:ASK}, when $r=\mu d\cdot \log(h_sk_s),$ there exists a one-to-one mapping between $[Q(\hat{g}_{\mathbf{R},U,t})]_2$ and ASK code,  say $\zeta_s: \{0,1\}^r\rightarrow [-\sqrt{P},\sqrt{P}]$. We therefore send the codeword $\varphi_{s}(\hat{g}_{t,\mathtt{shape}})=\zeta([Q(\tilde{g}_{\mathbf{R},U,t})]_2)$ in one channel-use as
\begin{align}
Y_{s,t}=\varphi_s(\hat{g}_{t,\mathtt{shape}})+Z_{s,t},\label{e:Gauss2}
\end{align} 
 where $Z_{s,t}\sim \cN(0,\sigma^2)$ denotes the Gaussian noise.  Note that the power constraint is always satisfied.

At the algorithm's end,  the decoder $\psi_s$ primarily makes use of three components, namely the minimum-distance decoder,  inverse mapping $\zeta_s^{-1}$, and inverse rotation, and performs the following steps:
\begin{enumerate}
\item The channel output $Y_{s,t}$ is fed into a minimum-distance decoder that locates the nearest possible ASK codeword in $[-\sqrt{P},\sqrt{P}]$ which further is fed into $\zeta_s^{-1}$ retrieving an $r$-bit sequence. 
\item The recovered $r$-bit sequence is split into blocks of size $\mu d\cdot \log(k_s)$ and $\mu d\cdot \log(h_s)$, each of which gets further split into sub-blocks of sizes $\log(k_s)$ and $\log(h_s)$, respectively.  These sub-blocks can uniquely identify the quantization intervals and the corresponding quantization levels for all sampled coordinates in $U$.  We denote that by $\hat{Q}(\tilde{g}_{\mathbf{R},U,t})\in \R^d$.  Note that all the remaining unsampled coordinates are decoded to be 0, i.e.,  
\[\hat{Q}(\tilde{g}_{\mathbf{R},U,t})(i)=0, \quad \forall i\notin U.\]
\item The last step is to multiply $\hat{Q}(\tilde{g}_{\mathbf{R},U,t})$ by $\frac{1}{\mu}$ and perform inverse rotation to get the decoded output $\psi_s(Y_{s,t})=\frac{1}{\mu}\mathbf{R}^{-1}\hat{Q}(\tilde{g}_{\mathbf{R},U,t}).$
\end{enumerate}

Under perfect minimum-distance decoding, we have $\hat{Q}(\tilde{g}_{\mathbf{R},U,t})=Q(\tilde{g}_{\mathbf{R},U,t}).$ Again, similar to \eqref{def_a}, \eqref{def_b} we have, 
\vspace{-0.25cm}
\begin{align*}
&\alpha(\pi,\varphi_s, \psi_s):=\sup_{\substack{\hat{g}_t \in \R^d: {\|\hat{g}_t\|^2}\leq 1}} \sqrt{\E {\| \psi_s(Y_{s,t})\|^2}}, \\
&\beta(\pi,\varphi_s, \psi_s):= \sup_{\substack{\hat{g}_t\in \R^d: {\|\hat{g}_t\|^2}\leq 1}}\|\E {\|\hat{g}_t\|-\psi_s(Y_{s,t})}\|,
\vspace{-0.25cm}
\end{align*}
where $Y_{s,t}$ is defined in \eqref{e:Gauss2}.
Since the shape vector is almost-surely bounded by $1$,  note that this time, the performance measures are defined over the class of almost-surely bounded oracles.

Following the proof of \cite[Theorem 3.7]{mayekar2020ratq}, we can derive lemma below.  
\begin{lem}\label{l:shape_perf}
For $\varphi_s, \psi_s$ as defined above and under perfect minimum-distance decoding event,  we have
\begin{align*}
\alpha(\pi, \varphi_s, \psi_s)
\leq \sqrt{\frac{1}{\mu} \left( \frac{9}{(k_s-1)^2}+1\right)} \hbox{ and } \beta(\pi, \varphi_s, \psi_s)=0.
\end{align*}
\end{lem}

\paragraph{3. Combining the gain and the shape quantizers.} The final decoded output is taken to be product of outputs from  gain decoder $\psi_g$ and shape decoder $\psi_s$,  i.e.,
\[\psi(Y_{g,t},Y_{s,t})=\psi_g(Y_{g,t})\cdot\psi_s(Y_{s,t}).\] The lemma below is again adapted from \cite[Theorem 4.2]{mayekar2020ratq} and can be proved in a similar way.
\begin{lem}\label{l:combined_perf}
For the decoded output $\psi(Y_{g,t},Y_{s,t})$ defined above, we have
\begin{align*}
\alpha(\pi, \varphi,\psi)&\leq \alpha(\pi, \varphi_g,\psi_g)\cdot \alpha(\pi, \varphi_s,\psi_s) \hbox{ and }\\
\beta(\pi, \varphi,\psi)&\leq \beta(\pi, \varphi_g,\psi_g).
\end{align*}
\end{lem}

\paragraph{4. Analysis.} Since communicating gain and shape  for each query requires 2 channel-uses, we can have atmost $N/2$ queries.  For the minimum-distance decoder,  denote by $A_N$ the event where all the ASK constellation points
sent in $N$ channel-uses
are decoded correctly by the algorithm and by $A_N^c$ its complement, i.e.,
$A_N^c := \cup_{t=1}^{N} \{|Z_{s,t}| \geq 2\sqrt{P}/(2^r-1)\}$,
where $Z_{s,t}$ is defined in \eqref{e:Gauss2}.
We have
\begin{align*}
\E{f(x_N)-f(x^\ast )}&=\E{(f(x_N)-f(x^\ast))\mid A_{N}}\cdot \bPr{A_{N}}+\E{(f(x_T)-f(x^\ast))\mid A_{N}^c}\cdot \bPr{A_{N}^c}\\
&\leq \E{(f(x_N)-f(x^\ast))\mid A_{N}}+DB\cdot \bPr{A_{N}^c}. \numberthis \label{e:ubs}
\end{align*}
As $\bPr{A_N^c} \leq N\exp\left( -\frac{2\SNR}{(2^r-1)^2}\right)$,  setting $r=\log\left( \sqrt{\frac{4\SNR}{\ln N}}+1\right)$ gives $\bPr{A_N^c}\leq \frac{1}{\sqrt{N}}.$ 
Using Lemma \ref{thm1}, the first term on the right-side can be bounded as 
\begin{align*}
\E{(f(x_N)-f(x^\ast))\mid A_{N}}&\leq D\left( \frac{\alpha(\pi, \varphi, \psi)}{\sqrt{N/2}}+\beta(\pi, \varphi, \psi)\right).
\end{align*}

We now analyse the overall performance measures $\alpha(\pi, \varphi, \psi)$ and $\beta(\pi, \varphi, \psi)$ of gain-shape quantizer described above.

{ Recall that the gain value is sent in one channel-use after appropriate scaling.}  The shape is quantized  using RATQ and sent over the channel using the ASK code given by \eqref{e:ASK} with $r=\mu d \log(h_s k_s)$.

Using the individual performance measures from \eqref{e:perf1} and Lemma \ref{l:shape_perf},  and combining them via Lemma \ref{l:combined_perf}, we have
\begin{align*}
D\left( \frac{\alpha(\pi, \varphi, \psi)}{\sqrt{N/2}}+\beta(\pi, \varphi, \psi)\right) &\leq D\left( \frac{B\sqrt{1+\frac{1}{\SNR}}\cdot \sqrt{\frac{1}{\mu}\left(\frac{9}{(k_s-1)^2}+1 \right)}}{\sqrt{N/2}}\right)\\
&=DB\sqrt{\frac{d}{N}}\sqrt{\frac{2}{\mu d}\left(1+\frac{1}{\SNR}\right)\left(\frac{9}{(k_s-1)^2}+1\right)}\\
&=DB\sqrt{\frac{d}{Nr}}\sqrt{2\log (h_sk_s)\left(1+\frac{1}{\SNR}\right)\left(\frac{9}{(k_s-1)^2}+1\right)},
\end{align*}
where the last line uses the fact that $\mu d=\lceil r/\log (h_s k_s)\rceil.$ Using the inequality above,  \eqref{e:ubs} can be further bounded as
 \begin{align*}
\E{f(x_N)-f(x^\ast )}&\leq {\frac{DB\sqrt{d}}{\sqrt{N}}\sqrt{\frac{2\log (h_sk_s)\left(1+\frac{1}{\SNR}\right)\left(\frac{9}{(k_s-1)^2}+1\right)}{\log(1+\sqrt{4\SNR/\ln N})}} +\frac{DB}{\sqrt{N}}.}
\end{align*}
At last, we use an 8-level shape quantizer for every coordinate, i.e., $k_s=8$, and choose the number of quantization intervals $h_s$   satisfying $\log h_s=\lceil[\log (1+\ln^\ast(d/3))\rceil$.  For $\ln^\ast (d/3) \leq 7,$ {we have $\mu d=\lceil r/6\rceil$, which further implies that $r\geq 6$, and that}
\begin{align*}
\E{f(x_N)-f(x^\ast )}
&\leq \frac{DB}{\sqrt{N}}\sqrt{\frac{d \cdot 24\left(1+\frac{1}{\SNR}\right)}{{\log(1+\sqrt{4\SNR/\ln N})}}} +\frac{DB}{\sqrt{N}}\\
&\leq \frac{2DB}{\sqrt{N}} \cdot \sqrt{\frac{d}{\min\lbrace d, \frac{{\log(1+\sqrt{4\SNR/\ln N})}}{24\left(1+\frac{1}{\SNR}\right)}\rbrace}}\\
&{\leq \frac{2DB}{\sqrt{N}} \cdot \sqrt{\frac{d}{\min\lbrace d, \frac{{\log(1+\sqrt{4\SNR/\ln N})}}{48}\rbrace}}},
\end{align*}
where the last line uses the inequality that $1+\frac{1}{\SNR}\leq 2$ for $\SNR>1$, which further holds since
$r\geq 6$.
\qed

\section{Experiments}\label{s:exp}
We evaluate the performance of our proposed analog and digital schemes ($c.f.$ Sections \ref{s:scheme1},  \ref{p:ASK_1}) which achieve the optimality of over-the-air optimization at low and high $\SNR$s, respectively.  Our experiments validate all our claims and are described below.

We consider the task of image classification and perform experiments on MNIST dataset, which has 60000 training and 10000 test samples.  In particular, the classifier for the MNIST dataset is implemented by training a 3-layer {\em Convolutional Neural Network} (CNN) that consists of a single convolution layer with 16 filters of dimension $3\times 3$ each and ReLU activation function,  followed by a $2\times 2$ max-pooling; one fully connected layer with dimensions $2704\times 10$; and a final softmax output layer, i.e., $d=27210$.  We choose the optimization algorithm to be SGD with learning rates proportional to $\SNR$s ({{as can be inferred from Lemma \ref{thm1}}}). 

{For our experiments,  we consider the proposed digital scheme using ASK described in Section \ref{p:ASK_1}.  Recall that the gain is always sent in one channel-use after scaling, and the shape is quantized using RATQ and then sent over the Gaussian channel using the ASK code. } For RATQ,  we set $h_s=4, k_s=8$ in tetra-iterated adaptive shape quantizer.  Further,  the descriptions of quantization interval ($\log h_s$ bits per dimension) and the corresponding uniform quantization point ($\log k_s$ bits per dimension) are sent separately\footnote{Note that our proposed digital scheme (see  Section \ref{p:ASK_1} for details) uses one channel transmission for sending the shape gradient quantization. Still, in the experiments, we send it using two channel transmissions.  We do this to mitigate the precision issues we run into for ASK coding at high values of $r$ in Python.} in two different channel uses, and the best values for $r$ in ASK code are chosen proportional to operating $\SNR$.  

On the other hand, for the proposed analog scheme,  we consider the sampled version of scaled transmission scheme described in Section \ref{s:panlgschm} with sampling only three coordinates, i.e., $\ell=3$. This choice of $\ell$ is considered for a fair performance comparison with the digital scheme in terms of the number of channel uses. Our codes are available online~\cite{OTAlink} on GitHub. 

{
We investigate the performance of the proposed analog and digital over-the-air schemes at various $\SNR$s{{; specifically,}} $-30$dB,  40dB, 100dB and 180dB.  We plot the training loss and test accuracy for the image classification task at these $\SNR$s in Figures~\ref{fig1},~\ref{fig2},~\ref{fig3} and~\ref{fig4}, respectively. {{The choice for these $\SNR$s are for illustrating the validity of theoretical claims, not for practical considerations. }}

The performance of both the analog and the digital scheme improves as $\SNR$ increases. However, {{the improvement  is faster for the digital scheme than for the analog scheme}}.  In more detail, Figure~\ref{fig1} shows that the proposed analog scheme performs much better than the proposed digital scheme at very low $\SNR$ of $-30$dB. 
As we increase the $\SNR$, the performance gap between the digital and analog schemes gets reduced.  
This can be observed in Figure~\ref{fig2} where the performance of both the schemes at $\SNR=40$dB is similar. 
With further increase in $\SNR$ values, the proposed digital scheme surpasses the performance of the proposed analog scheme, with the gap between their performance widening with an increase in $\SNR$, as can be seen in Figure~\ref{fig3} and Figure~\ref{fig4}.

Figure~\ref{fig4} also shows the performance of the classic {\em baseline} scheme, where perfect stochastic gradient estimates are available for the optimization protocol. In other words, the gradients are passed through a Gaussian channel of zero variance.
 As observed in Figure~\ref{fig4}, the proposed digital scheme is close to the baseline\footnote{In an ideal scenario, we expect the digital scheme to attain the baseline performance for a larger $\SNR$ value, as we increase the value of $r$ accordingly. Unfortunately, our python code runs into precision issues for optimally tuned ASK schemes at higher $\SNR$s (beyond $180$ dB).
} performance. 

Thus our experiments validate our theory. In particular, Figure~\ref{fig1} validates our theoretical claim that analog schemes are optimal at low $\SNR$. On the other hand, Figures~\ref{fig2},~\ref{fig3},~\ref{fig4} validates our theoretical claim that analog schemes go further away from optimality with an increase in $\SNR$ and digital schemes need to be used for optimal convergence at high $\SNR$.

\setcounter{figure}{0}   
  \pgfplotsset{compat=1.12, every axis/.append style={
  		line width=1pt, tick style={line width=0.6pt}}, width=7.6cm,height=5.5cm, tick label style={font=\small},
  	label style={font=\small},
  	legend style={font=\tiny}}

\begin{center}
\begin{figure}[h]
\begin{tikzpicture}
\begin{axis}[legend style={font=\footnotesize},
 	xlabel= Number of channel uses,
 	ylabel= Test accuracy,
    legend pos=north east,
 	xmin=0,
 	xmax=5900,
 	ymin=0,
 	ymax=1, grid=both, grid style={dotted, gray}
 	]
\addplot +[mark=triangle, color=ta2orange]table[x expr=351+351*\coordindex, y expr=\thisrowno{0}] {accrAnalogRCD.dat};
\addplot +[mark=star,color=cyan!60!taskyblue] table[x expr=351+351*\coordindex, y expr=\thisrowno{0}] {accrQlowSNR.dat};
\legend{ Proposed analog\\ Proposed digital\\}
\end{axis}
\end{tikzpicture}
\hspace{0.02cm}
\begin{tikzpicture}
\begin{axis}[legend style={font=\footnotesize},
 	xlabel= Number of channel uses,
 	ylabel= Training loss,
    legend pos=north east,
 	xmin=0,
 	xmax=5900,
 	ymin=1.5,
 	ymax=2.51, grid=both, grid style={dotted, gray}
 	]
\addplot +[mark=none, color=ta2orange]table[x expr=3*\coordindex, y expr=\thisrowno{0}] {trainlossAnalogRCD.dat};
\addplot +[mark=none,color=cyan!60!taskyblue] table[x expr=3*\coordindex, y expr=\thisrowno{0}] {trainlossQlowSNR.dat};

\legend{ Proposed analog\\ Proposed digital\\}
\end{axis}
\end{tikzpicture}
\caption{ Comparison between proposed analog and proposed digital scheme at $\SNR=-30$dB.}	
\label{fig1}
\end{figure}
\begin{figure}[h]
\begin{tikzpicture}
\begin{axis}[ 
legend style={font=\footnotesize},
 	xlabel= Number of channel uses,
 	ylabel= Test accuracy,
    legend pos=south east,
 	xmin=0,
 	xmax=5900,
 	ymin=0.5,
 	ymax=0.8, grid=both, grid style={dotted, gray}
 	]
 \addplot +[mark=triangle, color=ta2orange]table[x expr=351+351*\coordindex, y expr=\thisrowno{0}] {accrAnalogRotRCDSNR1e04.dat};
   \addplot +[mark=star, color=cyan!60!taskyblue] table[x expr=351+351*\coordindex, y expr=\thisrowno{0}] {accrSNRe04.dat};
\legend{ Proposed analog\\ Proposed digital\\}
\end{axis}
\end{tikzpicture}
\hspace{0.05cm}
\begin{tikzpicture}
\begin{axis}[
legend style={font=\footnotesize},
 	xlabel= Number of channel uses,
 	ylabel= Training loss,
    legend pos=north east,
 	xmin=0,
 	xmax=5900,
 	ymin=0.6,
 	ymax=2.51, grid=both, grid style={dotted, gray}
 	]
\addplot +[mark=none, color=ta2orange]table[x expr=3*\coordindex, y expr=\thisrowno{0}] {trainlossAnalogRotRCDSNR1e04.dat};
\addplot +[mark=none, color=cyan!60!taskyblue] table[x expr=3*\coordindex,y expr=\thisrowno{0}] {trainlossSNRe04.dat};
\legend{ Proposed analog\\ Proposed digital\\}
\end{axis}
\end{tikzpicture}
\caption{ Comparison between proposed analog and proposed digital scheme at $\SNR=40$dB.}		
\label{fig2}
\end{figure}
\begin{figure}[h]
 	\begin{tikzpicture}
 	\begin{axis}[legend style={font=\footnotesize},
 	xlabel= Number of channel uses,
 	ylabel= Test accuracy,
    legend pos=south east,
 	xmin=0,
 	xmax=5900,
 	ymin=0.6,
 	ymax=0.9,grid=both, grid style={dotted, gray} 
 	]

   \addplot +[color=cyan!60!taskyblue,mark=star] table[x expr=351+351*\coordindex, y expr=\thisrowno{0}] {accrSNRe10.dat};
 \addplot +[mark=triangle, color=ta2orange]table[x expr=351+351*\coordindex, y expr=\thisrowno{0}] {accrAnalogRotRCDSNR1e10.dat};

\legend{ Proposed digital\\ Proposed analog\\}
  \end{axis}
 	\end{tikzpicture}
\hspace{0.05cm}
 \begin{tikzpicture}
\begin{axis}[
legend style={font=\footnotesize},
 	xlabel= Number of channel uses,
 	ylabel= Training loss,
    legend pos=north east,
 	xmin=0,
 	xmax=5900,
 	ymin=0.6,
 	ymax=2.7, grid=both, grid style={dotted, gray}
 	]
\addplot +[mark=none, color=cyan!60!taskyblue] table[x expr=3*\coordindex, y expr=\thisrowno{0}] {trainlossSNRe10.dat};
\addplot +[mark=none, color=ta2orange]  table[x expr=3*\coordindex, y expr=\thisrowno{0}] {trainlossAnalogRotRCDSNR1e10.dat};

\legend{ Proposed digital\\ Proposed analog\\}
\end{axis}
\end{tikzpicture}
\caption{ Comparison between proposed analog and proposed digital scheme at $\SNR=100$dB.}	
\label{fig3}
\end{figure}
\begin{figure}[h]
\begin{tikzpicture}
\begin{axis}[legend style={font=\footnotesize},
 	xlabel= Number of channel uses,
 	ylabel= Test accuracy,
    legend pos=south east,
 	xmin=0,
 	xmax=5900,
 	ymin=0.5,
 	ymax=1, grid=both, grid style={dotted, gray}
 	]
 	\addplot +[color=blue, mark=none] table[x expr=117+117*\coordindex, y expr=\thisrowno{0}] {accrBL.dat};
\addplot +[mark=triangle, color=ta2orange]table[x expr=351+351*\coordindex, y expr=\thisrowno{0}] {accrAnalogRotRCDhighSNR.dat};
\addplot +[mark=star, color=cyan!60!taskyblue]table[x expr=351+351*\coordindex, y expr=\thisrowno{0}] {accrSNRe18.dat};
\legend{ Baseline\\ Proposed analog\\ Proposed digital\\}
\end{axis}
\end{tikzpicture}
\hspace{0.05cm}
\begin{tikzpicture}
\begin{axis}[legend style={font=\footnotesize},
 	xlabel= Number of channel uses,
 	ylabel= Training loss,
    legend pos=north east,
 	xmin=0,
 	xmax=5900,
 	ymin=0.2,
 	ymax=3, grid=both, grid style={dotted, gray}
 	]
\addplot +[mark=none, color=blue]table[x expr=\coordindex, y expr=\thisrowno{0}] {trainlossbl.dat};
\addplot +[mark=none, color=ta2orange]table[x expr=3*\coordindex, y expr=\thisrowno{0}] {trainlossAnalogRotRCDhighSNR.dat};
\addplot +[mark=none,color=cyan!60!taskyblue] table[x expr=3*\coordindex, y expr=\thisrowno{0}] {trainlossSNRe18.dat};

\legend{ Baseline\\ Proposed analog\\ Proposed digital\\}
\end{axis}
\end{tikzpicture}
\caption{ Comparison between proposed analog and proposed digital scheme at $\SNR=180$dB.}	
\label{fig4}
\end{figure}
\end{center}

\section{Concluding remarks}\label{s:conclusion}
We showed the optimality of analog schemes at low $\SNR$ in Corollary~\ref{c:AnOpt}. However, Theorem~\ref{t:ALB} shows that there is a $\sqrt{d}$ factor bottleneck that analog codes can't overcome, no matter how high the $\SNR$ is. Finally, we show in Theorem~\ref{t:ASK_1} that the proposed digital scheme using ASK codes almost attain the optimal convergence rate at all $\SNR$s.

It is important to note that more sophisticated coding schemes can still help in improving the small $\log\log N$ and $\ln^* d$ factors seen in the performance of ASK codes. 

In another direction, it is important to consider multiparty algorithms and multiterminal communication over Gaussian additive MAC channel. While the limitations for analog schemes apply to that setting as well, we may need to use lattice codes to extend our ASK coding scheme to a MAC. This is an interesting direction for future work.
\bibliography{tit2018}

\begin{thebibliography}{10}
\providecommand{\url}[1]{#1}
\csname url@samestyle\endcsname
\providecommand{\newblock}{\relax}
\providecommand{\bibinfo}[2]{#2}
\providecommand{\BIBentrySTDinterwordspacing}{\spaceskip=0pt\relax}
\providecommand{\BIBentryALTinterwordstretchfactor}{4}
\providecommand{\BIBentryALTinterwordspacing}{\spaceskip=\fontdimen2\font plus
\BIBentryALTinterwordstretchfactor\fontdimen3\font minus
  \fontdimen4\font\relax}
\providecommand{\BIBforeignlanguage}[2]{{%
\expandafter\ifx\csname l@#1\endcsname\relax
\typeout{** WARNING: IEEEtranS.bst: No hyphenation pattern has been}%
\typeout{** loaded for the language `#1'. Using the pattern for}%
\typeout{** the default language instead.}%
\else
\language=\csname l@#1\endcsname
\fi
#2}}
\providecommand{\BIBdecl}{\relax}
\BIBdecl

\bibitem{OTAlink}
Available online: \url{https://github.com/shubhamjha-46/OTA_Optimization}.

\bibitem{Abad20}
M.~S.~H. Abad, E.~Ozfatura, D.~Gündüz, and O.~Ercetin, ``Hierarchical
  {F}ederated {L}earning {ACROSS} {H}eterogeneous {C}ellular networks,'' in
  \emph{IEEE International Conference on Acoustics, Speech and Signal
  Processing (ICASSP)}, 2020, pp. 8866--8870.

\bibitem{acharya2020general}
J.~Acharya, C.~L. Canonne, Z.~Sun, and H.~Tyagi, ``Unified lower bounds for
  interactive high-dimensional estimation under information constraints,''
  \emph{\url{http://arxiv.org/abs/2010.06562v5}}, 2020.

\bibitem{ACT:18}
J.~Acharya, C.~L. Canonne, and H.~Tyagi, ``{Inference under Information
  Constraints {I}: Lower Bounds from Chi-Square Contraction},'' \emph{IEEE
  Transactions on Information Theory,}, 2020.

\bibitem{ach2020disc}
J.~Acharya, C.~L. Canonne, Y.~Liu, Z.~Sun, and H.~Tyagi, ``Interactive
  inference under information constraints,'' in \emph{Proceedings of the {IEEE}
  International Symposium of Information Theory (ISIT)}, 2021.

\bibitem{ACMT21informationconstrained}
J.~Acharya, C.~L. Canonne, P.~Mayekar, and H.~Tyagi, ``Information-constrained
  optimization: can adaptive processing of gradients help?''
  \emph{\url{https://arxiv.org/abs/2104.00979}}, 2021.

\bibitem{acharya2019distributed}
J.~Acharya, C.~De~Sa, D.~J. Foster, and K.~Sridharan, ``{Distributed Learning
  with Sublinear Communication},'' \emph{International Conference on Machine
  Learning}, 2019.

\bibitem{agarwal2012information}
A.~Agarwal, P.~L. Bartlett, P.~Ravikumar, and M.~J. Wainwright,
  ``{Information-Theoretic Lower Bounds on the Oracle Complexity of Stochastic
  Convex Optimization},'' \emph{IEEE Transactions on Information Theory},
  vol.~5, no.~58, pp. 3235--3249, 2012.

\bibitem{alistarh2017qsgd}
D.~Alistarh, D.~Grubic, J.~Li, R.~Tomioka, and M.~Vojnovic, ``{{QSGD}:
  Communication-efficient SGD via gradient quantization and encoding},''
  \emph{Advances in Neural Information Processing Systems}, pp. 1709--1720,
  2017.

\bibitem{alistarh2018convergence}
D.~Alistarh, T.~Hoefler, M.~Johansson, S.~Khirirat, N.~Konstantinov, and
  C.~Renggli, ``The convergence of sparsified gradient methods,''
  \emph{Advances in Neural Information Processing Systems}, 2018.

\bibitem{Amiric19}
M.~M. Amiri, T.~M. Duman, D.~Gündüz, S.~R. Kulkarni, and H.~Vincent~Poor,
  ``Collaborative {M}achine {L}earning at the {W}ireless {E}dge with {B}lind
  {T}ransmitters,'' \emph{IEEE Transactions on Wireless Communications}, pp.
  1--1, 2021.

\bibitem{Amiri19}
M.~M. Amiri and D.~Gündüz, ``Machine {L}earning at the {W}ireless {E}dge:
  {D}istributed {S}tochastic {G}radient {D}escent {O}ver-the-{A}ir,'' in
  \emph{IEEE International Symposium on Information Theory (ISIT)}, 2019, pp.
  1432--1436.

\bibitem{Amiria19}
------, ``Over-the-{A}ir {M}achine {L}earning at the {W}ireless {E}dge,'' in
  \emph{IEEE International Workshop on Signal Processing Advances in Wireless
  Communications (SPAWC)}, 2019, pp. 1--5.

\bibitem{Amirib20}
------, ``Federated {L}earning {O}ver {W}ireless {F}ading {C}hannels,''
  \emph{IEEE Transactions on Wireless Communications}, vol.~19, no.~5, pp.
  3546--3557, 2020.

\bibitem{basu2019qsparse}
D.~Basu, D.~Data, C.~Karakus, and S.~Diggavi, ``{Qsparse-local-{SGD}:
  Distributed {SGD} with Quantization, Sparsification, and Local
  Computations},'' \emph{Advances in Neural Information Processing Systems},
  2019.

\bibitem{bernstein18a}
J.~Bernstein, Y.-X. Wang, K.~Azizzadenesheli, and A.~Anandkumar, ``sign{SGD}:
  {C}ompressed {O}ptimisation for {N}on-{C}onvex {P}roblems,'' in
  \emph{Proceedings of the 35th International Conference on Machine Learning
  (ICML)}, vol.~80, 2018, pp. 560--569.

\bibitem{Chang20}
W.-T. Chang and R.~Tandon, ``Communication {E}fficient {F}ederated {L}earning
  over {M}ultiple {A}ccess {C}hannels,''
  \emph{\url{https://arxiv.org/abs/2001.08737}}, 2020.

\bibitem{Chen21}
M.~Chen, Z.~Yang, W.~Saad, C.~Yin, H.~V. Poor, and S.~Cui, ``A {J}oint
  {L}earning and {C}ommunications {F}ramework for {F}ederated {L}earning over
  {W}ireless {N}etworks,'' \emph{IEEE Transactions on Wireless Communications},
  vol.~20, no.~1, pp. 269--283, 2021.

\bibitem{chen2020breaking}
W.-N. Chen, P.~Kairouz, and A.~{\"O}zg{\"u}r, ``Breaking the
  communication-privacy-accuracy trilemma,'' \emph{Neural Information
  Processing Systems (NeurIPS)}, 2020.

\bibitem{CovTho06}
T.~M. Cover and J.~A. Thomas, \emph{Elements of Information Theory. 2nd
  edition}.\hskip 1em plus 0.5em minus 0.4em\relax John Wiley \& Sons Inc.,
  2006.

\bibitem{faghri2020adaptive}
F.~Faghri, I.~Tabrizian, I.~Markov, D.~Alistarh, D.~Roy, and
  A.~Ramezani-Kebrya, ``Adaptive gradient quantization for data-parallel sgd,''
  \emph{Advances in Neural Information Processing Systems}, 2020.

\bibitem{gandikota2019vqsgd}
V.~Gandikota, D.~Kane, R.~Kumar~Maity, and A.~Mazumdar, ``vqsgd: Vector
  quantized stochastic gradient descent,'' in \emph{Proceedings of The 24th
  International Conference on Artificial Intelligence and Statistics}, ser.
  Proceedings of Machine Learning Research.\hskip 1em plus 0.5em minus
  0.4em\relax PMLR, 2021, pp. 2197--2205.

\bibitem{gersho2012vector}
A.~Gersho and R.~M. Gray, \emph{Vector quantization and signal
  compression}.\hskip 1em plus 0.5em minus 0.4em\relax Springer Science \&
  Business Media, 2012, vol. 159.

\bibitem{ghosh2020distributed}
A.~Ghosh, R.~K. Maity, and A.~Mazumdar, ``Distributed newton can communicate
  less and resist byzantine workers,'' \emph{Advances in Neural Information
  Processing Systems}, 2020.

\bibitem{horadam2012hadamard}
K.~J. Horadam, \emph{Hadamard matrices and their applications}.\hskip 1em plus
  0.5em minus 0.4em\relax \hspace{-0.15cm}Princeton university press, 2012.

\bibitem{huang2019optimal}
Z.~Huang, W.~Yilei, K.~Yi \emph{et~al.}, ``Optimal sparsity-sensitive bounds
  for distributed mean estimation,'' \emph{Advances in Neural Information
  Processing Systems}, pp. 6371--6381, 2019.

\bibitem{jhunjhunwala2021adaptive}
D.~Jhunjhunwala, A.~Gadhikar, G.~Joshi, and Y.~C. Eldar, ``Adaptive
  quantization of model updates for communication-efficient federated
  learning,'' in \emph{IEEE International Conference on Acoustics, Speech and
  Signal Processing (ICASSP)}, 2021, pp. 3110--3114.

\bibitem{konevcny2016federated}
J.~Konečný, H.~B. McMahan, F.~X. Yu, P.~Richtarik, A.~T. Suresh, and
  D.~Bacon, ``Federated learning: Strategies for improving communication
  efficiency,'' \emph{NIPS Workshop on Private Multi-Party Machine Learning},
  2016.

\bibitem{lin2020achieving}
C.-Y. Lin, V.~Kostina, and B.~Hassibi, ``Differentially {Q}uantized {G}radient
  {D}escent,'' in \emph{IEEE International Symposium on Information Theory
  (ISIT)}, 2021.

\bibitem{mayekar2021wyner}
P.~Mayekar, A.~T. Suresh, and H.~Tyagi, ``Wyner-{Z}iv estimators: Efficient
  distributed mean estimation with side-information,'' in \emph{International
  Conference on Artificial Intelligence and Statistics}.\hskip 1em plus 0.5em
  minus 0.4em\relax PMLR, 2021, pp. 3502--3510.

\bibitem{mayekar2020limits}
P.~Mayekar and H.~Tyagi, ``Limits on gradient compression for stochastic
  optimization,'' \emph{Proceedings of the {IEEE} International Symposium of
  Information Theory (ISIT' 20)}, 2020.

\bibitem{mayekar2020ratq}
------, ``{RATQ}: A universal fixed-length quantizer for stochastic
  optimization,'' \emph{IEEE Transactions on Information Theory}, 2020.

\bibitem{nemirovski1995information}
A.~Nemirovsky, ``{Information-based complexity of convex programming},'' 1995,
  {Available Online } \url{http://www2.isye.gatech.edu/ne-mirovs/Lec_EMCO.pdf}.

\bibitem{saha2021decentralized}
R.~Saha, S.~Rini, M.~Rao, and A.~Goldsmith, ``Decentralized optimization over
  noisy, rate-constrained networks: How we agree by talking about how we
  disagree,'' in \emph{ICASSP 2021-2021 IEEE International Conference on
  Acoustics, Speech and Signal Processing (ICASSP)}.\hskip 1em plus 0.5em minus
  0.4em\relax IEEE, 2021, pp. 5055--5059.

\bibitem{seide20141}
F.~Seide, H.~Fu, J.~Droppo, G.~Li, and D.~Yu, ``1-bit stochastic gradient
  descent and its application to data-parallel distributed training of speech
  dnns,'' \emph{Fifteenth Annual Conference of the International Speech
  Communication Association}, 2014.

\bibitem{Sery19}
T.~Sery and K.~Cohen, ``A {S}equential {G}radient-{B}ased {M}ultiple {A}ccess
  for {D}istributed {L}earning over {F}ading {C}hannels,'' in \emph{57th Annual
  Allerton Conference on Communication, Control, and Computing (Allerton)},
  2019, pp. 303--307.

\bibitem{Kobi20}
------, ``On {A}nalog {G}radient {D}escent {L}earning {O}ver {M}ultiple
  {A}ccess {F}ading {C}hannels,'' \emph{IEEE Transactions on Signal
  Processing}, vol.~68, pp. 2897--2911, 2020.

\bibitem{Sery20}
T.~Sery, N.~Shlezinger, K.~Cohen, and Y.~C. Eldar, ``C{OTAF}: {C}onvergent
  {O}ver-the-{A}ir {F}ederated {L}earning,'' in \emph{IEEE Global
  Communications Conference (GLOBECOM)}, 2020, pp. 1--6.

\bibitem{Sun20}
Y.~Sun, S.~Zhou, and D.~Gündüz, ``Energy-{A}ware {A}nalog {A}ggregation for
  {F}ederated {L}earning with {R}edundant {D}ata,'' in \emph{IEEE International
  Conference on Communications (ICC)}, 2020, pp. 1--7.

\bibitem{suresh2017distributed}
A.~T. Suresh, F.~X. Yu, S.~Kumar, and H.~B. McMahan, ``{Distributed mean
  estimation with limited communication},'' \emph{Proceedings of the
  International Conference on Machine Learning (ICML' 17)}, vol.~70, pp.
  3329--3337, 2017.

\bibitem{wang2018atomo}
H.~Wang, S.~Sievert, S.~Liu, Z.~Charles, D.~Papailiopoulos, and S.~Wright,
  ``{Atomo: Communication-efficient learning via atomic sparsification},''
  \emph{Advances in Neural Information Processing Systems}, pp. 9850--9861,
  2018.

\bibitem{Wang18}
S.~Wang, T.~Tuor, T.~Salonidis, K.~K. Leung, C.~Makaya, T.~He, and K.~Chan,
  ``When {E}dge {M}eets {L}earning: {A}daptive {C}ontrol for
  {R}esource-{C}onstrained {D}istributed {M}achine {L}earning,'' in \emph{IEEE
  Conference on Computer Communications (INFOCOM)}, 2018, pp. 63--71.

\bibitem{wen2017terngrad}
W.~Wen, C.~Xu, F.~Yan, C.~Wu, Y.~Wang, Y.~Chen, and H.~Li, ``{{TernGrad}:
  Ternary gradients to reduce communication in distributed deep learning},''
  \emph{Advances in Neural Information Processing Systems}, pp. 1509--1519,
  2017.

\bibitem{Yang20}
K.~Yang, T.~Jiang, Y.~Shi, and Z.~Ding, ``Federated {L}earning via
  {O}ver-the-{A}ir {C}omputation,'' \emph{IEEE Transactions on Wireless
  Communications}, vol.~19, no.~3, pp. 2022--2035, 2020.

\bibitem{Zhang21}
J.~Zhang, N.~Li, and M.~Dedeoglu, ``Federated {L}earning over {W}ireless
  {N}etworks: A {B}and-limited {C}oordinated {D}escent {A}pproach,''
  \emph{\url{https://arxiv.org/abs/2102.07972}}, 2021.

\bibitem{Zhu21}
G.~Zhu, Y.~Du, D.~Gündüz, and K.~Huang, ``One-{B}it {O}ver-the-{A}ir
  {A}ggregation for {C}ommunication-{E}fficient {F}ederated {E}dge {L}earning:
  {D}esign and {C}onvergence {A}nalysis,'' \emph{IEEE Transactions on Wireless
  Communications}, vol.~20, no.~3, pp. 2120--2135, 2021.

\bibitem{Zhu20}
G.~Zhu, Y.~Wang, and K.~Huang, ``Broadband {A}nalog {A}ggregation for
  {L}ow-{L}atency {F}ederated {E}dge {L}earning,'' \emph{IEEE Transactions on
  Wireless Communications}, vol.~19, no.~1, pp. 491--506, 2020.

\end{thebibliography}
\bibliographystyle{IEEEtranS}
\appendix
\section{Mathematical details concerning Remark \ref{r:topk}}\label{ap:TopK}
Recall that in the top-$k$  gradient coding scheme only the absolute largest $k$ values of the gradients are used to update the query point. We begin by defining a strict generalization of top-$k$ gradient coding schemes which we call $k$-coordinate sampling codes.

\begin{defn}\label{d:analog_schemes}
A code is a $k$-coordinate sampling code if the encoder mapping $\varphi$ consist of only $k$-coordinate values and their indices, i.e., when $\varphi(x)= \left( S ,  \{x(i)\}_{i \in S}\right)$, where $S$ is a subset of $[d]$ with cardinality $k$.
 Further, we allow for the set $S$ to be dependent on $x$. Also, we denote by $\cE^\ast_{{\tt k}cs}(N)$ the min-max optimization error when the class of $(d, \ell, P)$-encoding protocol is restricted to analog schemes (with everything else remaining the same as in \eqref{e:minmax}). Clearly,  $\cE^\ast_{{\tt k}cs}(N) \geq \cE^\ast(N).$ 
\end{defn}
\begin{lem} For all values of $\SNR,$ we have
\[ \cE^\ast_{{\tt k}cs}(N) \geq  \frac{cDB}{\sqrt{N}}\cdot \sqrt{\frac{d}{\min\{d, k\log \frac{d}{k}\}}} \]
\end{lem}
\begin{proof}
For bounding  $ \cE^\ast_{{\tt k}cs}(N)$, our function class remains the same as in \eqref{eq:def:gv:convex:p2} and the oracle remain the same as in the proof of Theorem \ref{t:LB}. Now note that since gradient estimates supplied by the oracle are Bernoulli vectors, the encoder $\varphi(\cdot)$ can thought of as quantizer with precision of $\log{d \choose k} +k$ bits, where the first term in the addition is used to represent $S$ and the second to represent $k$. Therefore, even at infinite $\SNR$, we have
\[\sum_{i \in [d]}\mutualinfo{V(i)}{Y^T}\leq  c^{\prime}\delta^2 N \min \{d, \log{d \choose k }+k\}, \]
where the result directly follows from \cite[Theorem 5]{ACMT21informationconstrained}
Then, by noting that $\log{d \choose k} +k \leq k 
\log \frac{d}{k}+ k(1+\log e)$ and proceeding as in proof of Theorem \ref{t:LB}, the proof is complete.
\end{proof}

Thus, if we employ top-$k$ gradient coding schemes, even at very high $\SNR$ values we do not attain the classic convergence rate.
\end{document}